\documentclass[11pt, a4paper,reqno]{amsart}

\usepackage[T1]{fontenc}
\usepackage[utf8]{inputenc}
\usepackage{pgf,tikz,pgfplots}
\usepackage{mathrsfs}
\usetikzlibrary{arrows}

\usepackage{color} \definecolor{bleu_sombre}{rgb}{0,0,0.6}  \definecolor{rouge_sombre}{rgb}{0.8,0,0}\definecolor{vert_sombre}{rgb}{0,0.6,0}
\usepackage[plainpages=false,colorlinks,linkcolor=bleu_sombre,citecolor=rouge_sombre,urlcolor=vert_sombre,breaklinks]{hyperref}

\usepackage[english]{babel}

\usepackage{amsmath,amssymb,amsthm,graphicx,amsfonts,url,color,enumerate,dsfont,stmaryrd,mathabx, csquotes}

\theoremstyle{plain}
\newtheorem{theorem}{{Theorem}}[section] 
\newtheorem*{theorem*}{{Theorem}}
\newtheorem{proposition}[theorem]{Proposition}
\newtheorem*{proposition*}{Proposition}
\newtheorem{corollary}[theorem]{Corollary}
\newtheorem*{corollary*}{Corollary}

\newtheorem*{lemma*}{Lemma}

\theoremstyle{definition}
\newtheorem{definition}[theorem]{Assumption}
\newtheorem*{definition*}{Assumption}

\theoremstyle{remark}
\newtheorem{remark}[theorem]{Remark}

\theoremstyle{remark}

\makeatletter

\@addtoreset{equation}{section}  
\makeatother
\newcommand{\ab}[1]{\left| #1 \right|}

\newcommand{\dd}{\mathrm{d}}
\newcommand{\R}{\mathbb{R}}

\newcommand{\et}{{\boldsymbol{\tau}}}
\newcommand{\eg}{{\boldsymbol{\gamma}}}
\newcommand{\en}{{\boldsymbol{n}}}

\title[]{On the two-dimensional  quantum confined Stark effect  in  strong electric fields}

\author{H. D. Cornean}
\address[H. D. Cornean]{Department of Mathematical Sciences, Aalborg University, Skjernvej 4A, 9220 Aalborg, Denmark}
\email{cornean@math.aau.dk}

\author{D. Krej{\v{c}}i{\v{r}}{\'{\i}}k}
\address[D. Krej{\v{c}}i{\v{r}}{\'{\i}}k]
{Department of Mathematics, Faculty of Nuclear Sciences and
	Physical Engineering, Czech Technical University in Prague,
	Trojanova 13, 12000 Prague 2, Czechia}
\email{david.krejcirik@fjfi.cvut.cz}
\author{T. G. Pedersen}
\address[T. G. Pedersen]{Department of Materials and Production, Aalborg University, Skjernvej 4A, 9220 Aalborg, Denmark}
\email{tgp@mp.aau.dk}

\author{N. Raymond}
\address[N. Raymond]{Laboratoire Angevin de Recherche en Mathématiques, LAREMA, UMR 6093, UNIV Angers, SFR Math-STIC, 2 boulevard Lavoisier 49045 Angers Cedex 01, France}
\email{nicolas.raymond@univ-angers.fr}

\author[E. Stockmeyer]{E. Stockmeyer}

\address[E. Stockmeyer]{Instituto de F\'isica, Pontificia Universidad Cat\'olica de Chile, Vicu\~na Mackenna 4860, Santiago 7820436, Chile.}
\email{stock@fis.puc.cl}
\begin{document}

	\begin{abstract}
 We consider a 
	Stark Hamiltonian on a two-dimensional bounded domain with Dirichlet boundary conditions. In the strong
	electric field limit  we derive, under certain local convexity conditions, a three-term asymptotic expansion of the
	low-lying eigenvalues. This shows that the excitation frequencies are proportional to the square root of the boundary curvature at a certain point determined by the direction of the electric field.
\end{abstract}
\maketitle	
\section{Introduction}

\subsection{Physical motivation}
There has  been a lot of theoretical and experimental work on the
properties of quantum confined semiconductor devices. These systems exhibit  interesting features that may find applications in  
 nano-technology (see
e.g. \cite{harrison2016quantum,HARUTYUNYAN2009695}). Physically, the
confinement  may be achieved through an interface with the vacuum or an hetero-junction, i.e., by
interfacing the semiconductor with an isolator or with another
semiconductor of larger gap  
\cite{harrison2016quantum}. The most prominent semiconductor devices
are the so-called quantum wells, quantum wires and quantum dots, where
the material is confined in one, two and three orthogonal directions,
respectively. In the presence of symmetry along the non-confined
directions the description of the system is reduced to the study of a
Hamiltonian in dimension one, for quantum wells, and  in dimension two for quantum
wires (see e.g., \cite[Chapter 8]{harrison2016quantum}).

An interesting  phenomenon is the behavior of the energy levels of the
system when a uniform electric field is applied; this is known as
the quantum confined Stark effect \cite{PRL2173} (see also
\cite{PhysRevA.99.063410} and references therein).  In a simplified
model one may consider independently electrons or
holes. Then, in the so-called effective mass envelope function
approximation \cite{doi:10.1063/1.99562}, the effective Hamiltonian
describing the quantum confined Stark effect is formally given by
\begin{align}
\label{eq:1}
-\frac{\hbar^2}{2m} \Delta +q\,{\mathbf { F}}\cdot {\mathbf x}+V_{\rm conf}\,,
\end{align}
where $\hbar$ is Plank's constant divided by $2\pi$, $m>0$ is the
effective mass of the electron (or hole), $-q$ is the charge of the
particle (electron or hole), ${\mathbf  F}$ is the electric field, and
$V_{\rm conf}$ is some confining potential. In the simplest case one
models the confining potential as infinite potential walls, i.e., one considers
the first two terms of the Hamiltonian in \eqref{eq:1} restricted to the
domain of confinement with Dirichlet boundary conditions. The
quantum confined Stark effect modeled  as in \eqref{eq:1}
with Dirichlet boundary conditions have been considered, for instance,
in \cite{PhysRevB.28.3241,PRL2173} for quantum wells, in
\cite{PhysRevB.41.12911, Wang_2013,
	doi:10.1002/pssb.200743045,doi:10.1002/pssb.200301865,doi:10.1063/1.1849430,HARUTYUNYAN2009695}
for quantum wires, and in \cite{WEI2008786,PhysRevA.99.063410} for
quantum dots.

In this work we are interested in the study of the low-lying
eigenvalues of the following
two-dimensional Hamiltonian restricted to an open set $\Omega\subset\R^2$:
\begin{align}
\label{eq:2}
H=-\frac{\hbar^2}{2m}(\partial_x^2+\partial_y^2)+q{F}x\,,
\end{align}
acting on a dense subspace of the square integral functions
$L^2(\Omega)$ with Dirichlet boundary conditions.  This is a model
Hamiltonian to describe the energy levels of a quantum wire with cross
section $\Omega$ in the presence of an electric field perpendicular to
the non-confined direction. Notice that we have chosen coordinates
such that ${\mathbf F }$ is parallel to the $x$-axis.

The low-lying eigenvalues of this operator have been studied, partially numerically, for
different geometries in several papers such as in
\cite{PhysRevB.41.12911, Wang_2013} for squares, in
\cite{doi:10.1002/pssb.200743045} for rectangles, in
\cite{doi:10.1002/pssb.200301865, doi:10.1063/1.1849430, Pedersen_2017} for disks,
and in \cite{HARUTYUNYAN2009695} for annuli. 

From now on we will work with domains satisfying the following conditions, see Figure~\ref{Fig:1}.
\begin{definition}\label{ass1}
The set $\Omega$ is open, bounded, and connected. We assume that there exists a unique point $A_0\in \partial\Omega$ such that the first component of $A_0$ is given by 
\begin{align*}
x_{\min}=\inf_{(x,y)\in\Omega}x=\min_{(x,y)\in\overline{\Omega}}x\,.
\end{align*}
We also assume that $\partial\Omega$ is smooth near $A_0$, and that the curvature at $A_0$, denoted by $\kappa_0$, is positive.
\end{definition}
Let us describe the notion of curvature we use in this paper. Consider a smooth counterclockwise parametrization $s\mapsto\gamma(s)$ by arc-length of the boundary near $A_0=\gamma(0)$. If $\en(s)$ is the outward pointing normal to $\partial\Omega$ at $\gamma(s)$, the curvature $\kappa(s)$ at $\gamma(s)$ is defined through the relation $\gamma''(s)=-\kappa(s)\en(s)$. We have $\kappa(0)=\kappa_0>0$.

\begin{figure}[ht]
\centering
\includegraphics[width=12cm, height=9cm]{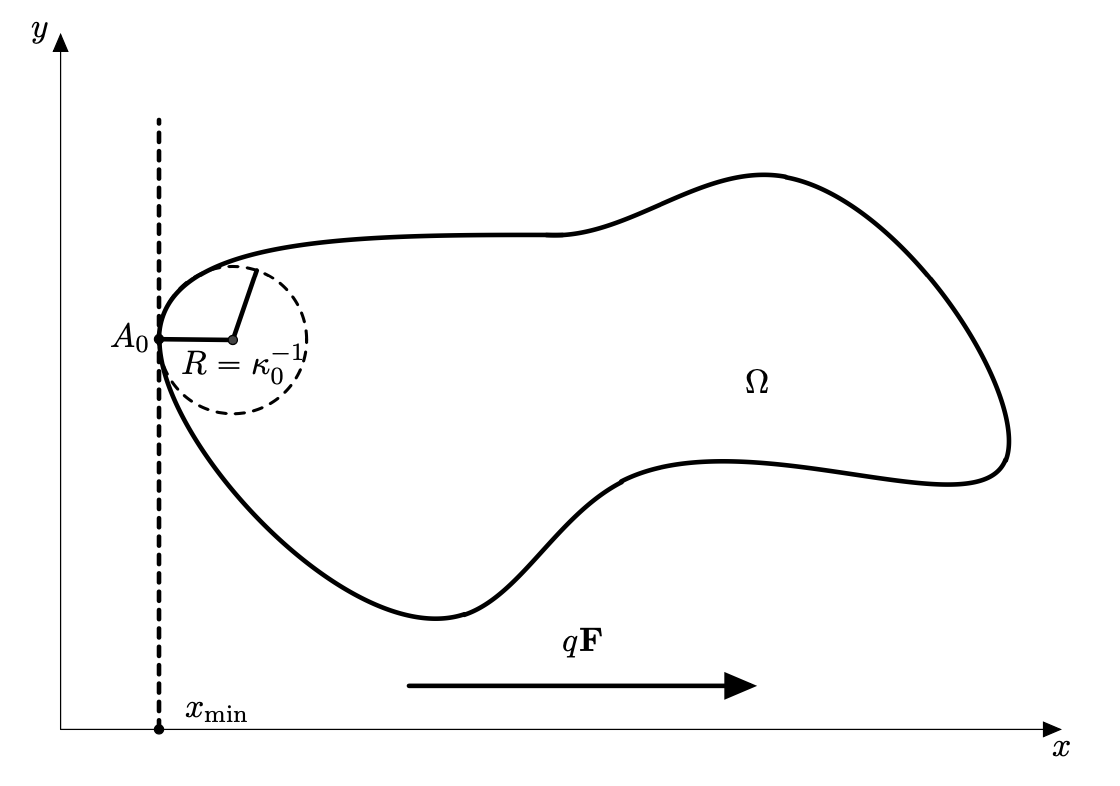}
\caption{The domain $\Omega$.}
\label{Fig:1}
\end{figure}

Then we can provide a three
terms asymptotic expansion for the individual eigenvalues of $H$ in
the limit of strong electric fields.  Our main result, Theorem
\ref{thm.main}, implies that for $q F >0$ the $n$-th eigenvalue of $H$
behaves in the limit of strong electric field as
\begin{align}
\label{eq:3}
E_n=  q F x_{\rm min}+\left( \frac{(q F \hbar)^2}{2m}\right)^{\frac{1}{3}}z_1
+(n-\tfrac{1}{2}) \hbar \left( \frac{q F \kappa_0 }{m}\right)^{\frac{1}{2}}
+ \mathscr{O}(F^{1/3})\,,
\end{align}
where $z_1\approx 2.338$ is the absolute value of the smallest zero of the
Airy function.

In particular, we find an
interesting behaviour of the energy splitting in terms of the
geometry of $\Omega$  which is proportional to the $F^{1/2}$. Such an equidistant splitting for the eigenvalues in the  strong electric field regime  has been observed in
\cite{doi:10.1063/1.1849430} for disk shaped $\Omega$ using
partially numerical methods (see Figure 7 from
\cite{doi:10.1063/1.1849430}). However, its dependence on the
curvature does not seem to have been reported before. 
Heuristically the third term in the expansion may be explained as follows:  under a strong electric field  the particle is pushed towards $x_{\rm min}$ and behaves as an harmonic oscillator in the direction perpendicular to the field with elastic constant  proportional to the curvature at $x_{\rm min}$. Let us mention here that a physics-oriented paper is in preparation. Our first investigations in this direction suggest that the spectral splitting \eqref{eq:3} might be experimentally accessible.

\subsection{Main result}
Let $\Omega$ be  as in Assumption \ref{ass1}. Notice that by factorizing $qF>0$ in the expression of $H$ in \eqref{eq:2} we have 
\begin{equation}\label{e-2}
\mathscr{L}_h:=(qF)^{-1}H=-h^2(\partial^2_x+\partial^2_y)+x\,,\quad h=\frac{\hbar}{\sqrt{2mqF}}\,.
\end{equation}
 We define 
$\mathscr{L}_h$ as the unique self-adjoint operator defined through  the quadratic form 
$$
\mathscr{Q}_h(\varphi)= h^2\int_{\Omega}\ab{\nabla\varphi(x)}^2 \dd x+\int_{\Omega} x\ab{\varphi(x)}^2 \dd x\,,\quad \varphi\in H^1_0(\Omega)\,.
$$
The operator $\mathscr{L}_h$  has domain contained in $H^1_0(\Omega)$ and acts as in \eqref{e-2}. This is the Dirichlet realization of the Stark Hamiltonian. We want to describe the first eigenvalues of this operator in the limit $h\to 0$. We obtain \eqref{eq:3} from the following result.
Denote by $(\lambda_n(h))_{n\geq 1}$ the  eigenvalues of $\mathscr{L}_h$ in increasing order, where each eigenvalue is repeated according to its multiplicity.
\begin{theorem}\label{thm.main}
Let $n\in \{1,2,\dots\}$. Then, we have as $h\to 0$
\[\lambda_n(h)=x_{\min}+z_{1}h^{\frac 23}+h(2n-1)\sqrt{\frac{\kappa_0}{2}}+\mathscr{O}(h^{ \frac 43})\,.\]
\end{theorem}

\begin{remark}
Various extensions of our main theorem can be considered.
\begin{enumerate}[\rm i.]
	\item It is possible to prove asymptotic expansions in powers of $h^{1/6}$ of the low-lying eigenvalues by using a formal series analysis.
	\item In our generic geometric situation, we could even prove that the eigenfunctions admit WKB expansions.
	\item Our strategy may be adapted to deal with a finite number of non-degenerate minima, and it would even be possible to investigate tunnel effects when these minima have symmetries.
\end{enumerate}
The proofs of such extensions can be adapted from a recent literature developed in the context of the Born-Oppenheimer approximation (see for instance \cite[Chapter 11]{Ray17}, or the generalizations \cite{Keraval, Martinez07}, and also \cite{HKR17} where tunnelling estimates are also considered). 
\end{remark}
\begin{remark}
	If we replace  the potential $x$ by $ix$ in the Hamiltonian and one looks at the real part of the eigenvalues the second term in the asymptotic expansion has a factor $1/2$ (see e.g., 
	\cite{Grebenkov2017TheCA,Henry2014OnTS,MR3542006,MR2438788}). Concerning the analysis in the case of imaginary electric fields, the reader might also want to consider \cite{GH18, AGH18, AGH19}. 
\end{remark}
The rest of this article is organized as follows: 
In Section \ref{sec2} we show that low energy eigenfunctions and its derivatives are exponentially well localized around $x_{\min}$.  We use this to find an effective Hamiltonian $\widetilde{\mathscr{M}}_h$, whose low energy eigenvalues are those of $H$  modulo an exponentially small error, this is done in Section \ref{sec3}. This operator is  expressed in tubular coordinates and acts on a tiny domain around $x_{\rm min}$ with Dirichlet boundary conditions.  In the last section we provide asymptotic upper and lower bounds for the eigenvalues of  $\widetilde{\mathscr{M}}_h$.

\section{Localization near the potential minimum}\label{sec2}
The following proposition states that the eigenfunctions associated with the low-lying eigenvalues are localized in $x$ near $x_{\min}$.
\begin{proposition}\label{thm.Agmon}
Let $M>0$. There exist $\varepsilon, C,h_0>0$ such that, for all $h\in(0,h_0)$, for all eigenvalues $\lambda$ such that $\lambda\leq x_{\min}+Mh^{\frac 23}$, and all corresponding eigenfunctions $\psi$,
\begin{equation}\label{hc-1}
\int_{\Omega}e^{\varepsilon|x-x_{\min}|^{\frac 32}/h}|\psi|^2\dd {\bf x}\leq C\|\psi\|^2\,,
\end{equation}
and
\begin{equation}\label{hc-2}
\int_{\Omega}e^{\varepsilon|x-x_{\min}|^{\frac 32}/h}|h\nabla\psi|^2\dd {\bf x}\leq Ch^{2/3}\|\psi\|^2\,.
\end{equation}
\end{proposition}

\begin{proof}
We write the Agmon formula, for all $\psi\in\mathrm{Dom}(\mathscr{L}_h)$ and all bounded Lipschitz functions $\Phi$:
\[\langle \mathscr{L}_h\psi,e^{2\Phi/h}\psi\rangle=\mathscr{Q}_h(e^{\Phi/h}\psi)-\|e^{\Phi/h}\nabla\Phi\psi\|^2\,.\]
 Let $\psi$ be an eigenfunction corresponding to an eigenvalue $\lambda\leq x_{\min}+Mh^{\frac 23}$. We get
\begin{equation}\label{eq.agmonidentity}
\int_{\Omega}h^2|\nabla e^{\Phi/h}\psi|^2+(x-\lambda-|\nabla\Phi|^2)|e^{\Phi/h}\psi|^2\dd x\dd y= 0\,,
\end{equation}
and thus
\begin{align}\label{e-1}
\int_{\Omega}h^2|\nabla e^{\Phi/h}\psi|^2+(x-x_{\min}-Mh^{2/3}-|\nabla\Phi|^2)|e^{\Phi/h}\psi|^2\dd x\dd y\leq 0\,.
\end{align}
Now, we choose $\Phi=\varepsilon |x-x_{\min}|^{3/2}$ and drop the first term above to get that
\[\int_{\Omega}\left(\left(1-\frac{9}{4}\varepsilon^2\right)(x-x_{\min})-Mh^{2/3}\right)|e^{\Phi/h}\psi|^2\dd x\dd y\leq 0\,.\]
We take $\varepsilon$ such that $1-\frac{9}{4}\varepsilon^2>0$ and
fix $R>0$ such that
\begin{equation}\label{hc-3}
\left(1-\frac{9}{4}\varepsilon^2\right)R-M=1\,.
\end{equation}
We write
\begin{multline*}
\int_{|x-x_{\min}|\geq Rh^{2/3}}\left(\left(1-\frac{9}{4}\varepsilon^2\right)(x-x_{\min})-Mh^{2/3}\right)|e^{\Phi/h}\psi|^2\dd x\dd y\\
\leq-\int_{|x-x_{\min}|< Rh^{2/3}}\left(\left(1-\frac{9}{4}\varepsilon^2\right)(x-x_{\min})-Mh^{2/3}\right)|e^{\Phi/h}\psi|^2\dd x\dd y \,,
\end{multline*}
and get (using \eqref{hc-3}):
\begin{align*}
&h^{2/3}\int_{|x-x_{\min}|\geq Rh^{2/3}}|e^{\Phi/h}\psi|^2\dd x\dd y\\
&\leq \int_{|x-x_{\rm min}|\ge Rh^{2/3}}
\left(\left(1-\frac{9}{4}\varepsilon^2\right)(x-x_{\min})-Mh^{2/3}\right)|e^{\Phi/h}\psi|^2\dd x\dd y\\
&\leq Mh^{2/3}\int_{|x-x_{\min}|< Rh^{2/3}} e^{2\Phi/h}|\psi|^2\dd x\dd y \leq Ch^{2/3}\|\psi\|^2\,.
\end{align*}
This proves \eqref{hc-1}. Next we  prove \eqref{hc-2}.
 	First observe that repeating the calculation above without dropping the first term in 
 	\eqref{e-1} we get that
 	\begin{align*}
 	\int_{\Omega}h^2|\nabla e^{\Phi/h}\psi|^2\dd x\dd y\leq  Ch^{2/3}\|\psi\|^2\,.
 	\end{align*} 	
 	 Next, fix  $0<\tilde{\varepsilon}<\varepsilon$. We have that
$$e^{\tilde{\varepsilon}\vert x-x_{\rm min}\vert^{\frac{3}{2}}/h}(h\nabla \psi)=-\frac{3\tilde{\varepsilon}}{2}\sqrt{ x-x_{\rm min}}\; e^{\tilde{\varepsilon}\vert x-x_{\rm min}\vert^{\frac{3}{2}}/h} \psi+h\nabla \left (e^{\tilde{\varepsilon}\vert x-x_{\rm min}\vert^{\frac{3}{2}}/h} \psi\right ).$$
Using the triangle inequality,  \eqref{eq.agmonidentity}, \eqref{hc-1} and the estimate 
$$\sup_{t\geq 0}\sqrt{t}\; e^{-(\varepsilon-\tilde{\varepsilon})t^{\frac{3}{2}}/h}\leq C h^{\frac{1}{3}},$$
we complete the proof of \eqref{hc-2}.
\end{proof}
We denote by $\complement B(A_0,\eta)$ the complement of the open disc $B(A_0,\eta)$.
\begin{corollary}\label{cor-hc}
	Let $M,\eta >0$. There exist $c_\eta,C_\eta, h_0>0$ such that if  $h\in(0,h_0)$, then any eigenfunction $\psi$ corresponding to an eigenvalue $\lambda\leq x_{\min}+Mh^{\frac 23}$  satisfies the estimates
	\[\int_{\Omega\cap\complement B(A_0,\eta)}|\psi|^2\dd {\bf x}\leq C_\eta e^{-c_\eta/h}\|\psi\|^2\,\]
	and
	\[\int_{\Omega\cap\complement B(A_0,\eta)}|\nabla\psi|^2\dd {\bf x}\leq C_\eta e^{-c_\eta/h}\|\psi\|^2\,.\]
\end{corollary}

\begin{proof}
The set $\overline{\Omega}\cap\complement B(A_0,\eta)$ is compact. The map $\overline{\Omega}\cap\complement B(A_0,\eta)\ni  (x,y)\mapsto x-x_{\min}$ is continuous  and positive, thus it has a positive lower bound. 
The conclusion follows from Theorem \ref{thm.Agmon}.
\end{proof}

\section{Tubular coordinates and localized operator}\label{sec3}

We can now reduce our investigation to a neighborhood of $A_0$. 

\subsection{Tubular coordinates}
\begin{figure}[ht]
\centering
\includegraphics[width=7cm, height=10cm]{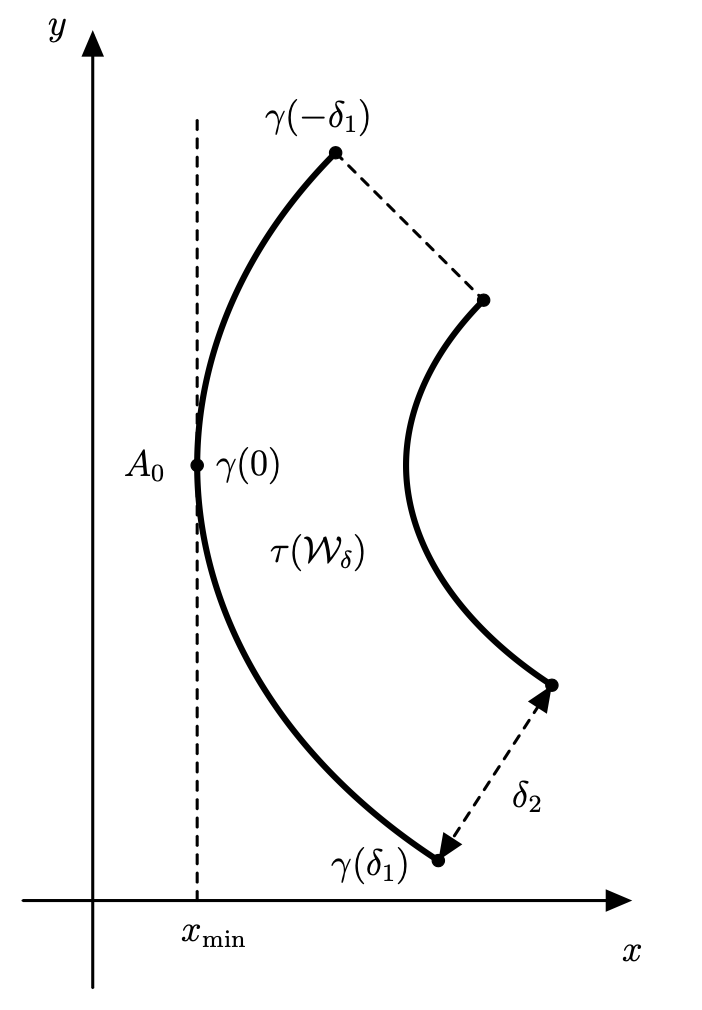}
\caption{The local tubular coordinates}
\label{Fig:2}
\end{figure}
 We use tubular coordinates in a neighborhood of  $A_0$ (see for instance \cite{FH11}). Due to Assumption \ref{ass1} there exist $\delta_1, \delta_2>0$ such that (see Figure \ref{Fig:2})
\[\mathcal{W}_\delta=(-\delta_1,\delta_1)\times (0,\delta_2)\ni (s,t)\mapsto \et(s,t)=\eg(s)-t\en(s)\,\in \et(\mathcal{W}_\delta)\]
induces a diffeomorphism. Here, $\en$ is the outward pointing normal and $\eg$ is the  natural length-parametrization of the boundary $\partial \Omega$; we set  $\eg(0)=A_0$. 
Denoting by $\theta(s)$  the turning angle at the point $\eg(s)\in \partial\Omega$ 
we may write $\en(s)=(\cos\theta(s), \sin\theta(s))$,  the tangent vector $\eg'(s)=(- \sin\theta(s), \cos\theta(s))$,  and the curvature $\kappa(s)=\theta'(s)$.  
The Jacobian of $\et=(\tau_1,\tau_2)=\tau_1{\bf e}_1+\tau_2{\bf e}_2$ is given by 
\begin{align}\label{es-2}
m(s,t)=1-\kappa(s) t\,.
\end{align}
We fix  $\delta_2>0$ so small that the map $\et$ induces a local diffeomorphism between a rectangle  and the tubular neighborhood of the boundary. 
In view of the definition of $A_0$ we have
\begin{align}\label{e-3}
\en'(0)\cdot {\bf e}_1 =0\,,\quad\eg'(0)\cdot {\bf e}_1 =0\,,\quad\eg''(0)=-\kappa_0\en(0)=\kappa_0{\bf e}_1\,,
\end{align}
where $\kappa _0=\kappa(0)$.

\subsection{Spectral reduction to a localized operator}
For $0<\delta\le \delta_1,\delta_2 $ we define $\mathscr{L}_{h,\delta}$ to be  the Dirichlet realization of $-h^2\partial^2_x-h^2\partial^2_y+x$ on $L^2(\et(\mathcal{W}_\delta))$ with $\mathcal{W}_\delta=(-\delta,\delta)\times (0,\delta)$. 

We denote by $\lambda_n(h,\delta)$ the corresponding eigenvalues of 
$\mathscr{L}_{h,\delta}$.
By using the decay estimates of Theorem \ref{thm.Agmon}, we get the following:
\begin{proposition}\label{prop.exp-approx}
	For every fixed $n\geq 1$ there exist $k,K>0$ such that
	\[\lambda_n(h,\delta)-Ke^{-k/h}\leq \lambda_n(h)\leq\lambda_n(h,\delta)\,.\]
\end{proposition}
\begin{proof}
The second inequality is a direct consequence of the min-max variational principle since the form domain of $\mathscr{L}_{h,\delta}$ is included in the one of $\mathscr{L}_{h}$. 

We now prove the first inequality but only for $n=1$. Let $\psi_1$ be a unit eigenvector of $\mathscr{L}_{h}$ corresponding to its groundstate $\lambda_1(h)$. Let $0\leq \chi_\delta\leq 1$ be a smooth cut-off function such that $\chi_\delta(\et(s,t))=1$ if $|s|,t\leq \delta/2$, and $\chi_\delta(\et(s,t))=0$ if $|s|\geq\frac34\delta$ or $t\geq\frac34\delta$. Then $\chi_\delta \psi_1$ belongs to the domain of $\mathscr{L}_{h,\delta}$ and since the support of the derivatives of $\chi_\delta$ is away from $A_0$, from Corollary \ref{cor-hc} we have 
$$\Vert \chi_\delta \psi_1\Vert =1+\mathscr{O}(e^{-k/h}),\; \Vert (\nabla\chi_\delta)\nabla \psi_1\Vert =\mathscr{O}(e^{-k/h}),\; \Vert (\mathscr{L}_{h,\delta}-\lambda_1(h))\chi_\delta \psi_1\Vert=\mathscr{O}(e^{-k/h}).$$
Then the min-max principle implies:
$$\lambda_1(h,\delta)\leq \frac{\langle\chi_\delta\psi_1, \mathscr{L}_{h,\delta}\chi_\delta\psi_1\rangle}{\Vert \chi_\delta \psi_1\Vert^2}\leq \lambda_1(h)+\mathscr{O}(e^{-k/h}).$$
If $n>1$, one can construct $n$ quasi-modes for $\mathscr{L}_{h,\delta}$ out of the first $n$ eigenmodes of $\mathscr{L}_{h}$ by using a similar argument.
\end{proof}

Therefore, we can focus on the spectral analysis of $\mathscr{L}_{h,\delta}$. The operator $\mathscr{L}_{h,\delta}$ is unitarily equivalent to the Dirichlet realization of 
\[\mathscr{M}_h=-h^2m^{-1}\partial_sm^{-1}\partial_s-h^2m^{-1}\partial_t m\partial_t+\tau_1(s,t)\,,\]
acting on $L^2(\mathcal{W}_\delta, m(s,t)\dd s\dd t)$.

By Taylor expansion near $(0,0)$,  we have
\begin{align*}
\et(s,t)=\eg(0)+s\eg'(0)+\frac{s^2}{2}\eg''(0)-t(\en(0)+s\en'(0))+\mathscr{O}(|s|^3+|t s^2 |)
\end{align*}
Thus, in view of \eqref{e-3}, we have in particular that
\begin{equation}\label{eq.c1}
\tau_1(s,t)=x_{\min}+\frac{\kappa_0}{2}s^2+t+\mathscr{O}(|s|^3+|ts^2|)\,.
\end{equation}

Note that for $\delta$ small enough, there exists $\alpha>0$ such that for all $(s,t)\in(-\delta,\delta)\times(0,\delta)$:
\[\tau_1(s,t)\geq x_{\min}+\alpha (s^2+t)\,.\]
\begin{proposition}\label{prop.Agmont}
	Let $M>0$. There exist $\varepsilon, C,h_0>0$ such that, for all $h\in(0,h_0)$, and  for all eigenfunctions $\psi$ corresponding to  eigenvalues $\lambda$ of $\mathscr{M}_h$ with $\lambda\leq x_{\min}+Mh^{\frac 23}$, we have:
	\begin{equation}\label{eq.Agmont1}
	\int_{\mathcal{W}_\delta}e^{\varepsilon t^{\frac 32}/h}|\psi|^2\dd s\dd t\leq C\|\psi\|^2\,,
	\end{equation}
	\begin{equation}\label{eq.Agmont2}
	\int_{\mathcal{W}_\delta}e^{\varepsilon t^{\frac 32}/h}|h\nabla_{s,t}\psi|^2\dd s\dd t\leq Ch^{2/3}\|\psi\|^2\,.
	\end{equation}
\end{proposition}
\begin{proof} 
If $\delta$ is small enough then $m(s,t)\sim 1$ and $x-x_{\rm min}\sim t$. Therefore we can directly use the strategy of  Proposition \ref{thm.Agmon} applied for the eigenfunctions of $\mathscr{L}_{h,\delta}$, but in both \eqref{eq.Agmont1} and \eqref{eq.Agmont2} we need to choose an $\varepsilon$ which is smaller than the one in  Proposition \ref{thm.Agmon} in order to control the linear growth in $t$ of $m(s,t)-1$ and $\partial_s m(s,t)$. 
\end{proof}

Therefore, the operator $\mathscr{M}_h$ can be replaced by 
\[\widetilde{\mathscr{M}}_h=-h^2m^{-1}\partial_sm^{-1}\partial_s-h^2m^{-1}\partial_t m\partial_t+\tau_1(s,t)\,,\]
with Dirichlet boundary conditions, acting on $L^2(\mathcal{W}_{\delta,h}, m(s,t)\dd s\dd t)$, with \[\mathcal{W}_{\delta,h}=(-\delta,\delta)\times(0,h^{\frac 23-\eta})\,,\] for some $\eta\in\left(0,\frac 13\right)$ with corresponding quadratic form
\begin{align}\label{qf-m-tilde}
\langle\psi, \widetilde{\mathscr{M}}_h\psi\rangle=
\int_{\mathcal{W}_{\delta,h}}\Big(m^{-2}|h\partial_s\psi|^2+
	|h\partial_t \psi|^2+\tau_1(s,t)|\psi|^2\Big)m\dd s\dd t\,.
\end{align} 

Let $\mu_n(h)$ be the asociated (ordered) eigenvalues of $\widetilde{\mathscr{M}}_h$. The decay estimates of Proposition \ref{prop.Agmont} are still satisfied by the eigenfunctions of $\widetilde{\mathscr{M}}_h$ with eigenvalues $\lambda\leq x_{\min}+Mh^{\frac 23}$. By using this exponential decay, there exist $C, h_0>0$ such that, for all $h\in(0,h_0)$,
\begin{equation}\label{eeq1}
\mu_n(h)-\tfrac{1}{C}e^{-Ch^{-3\eta/2}}\leq\lambda_n(h,\delta)\leq \mu_n(h)\,.
\end{equation}
Thus, modulo an exponentially small error, the asymptotic analysis of $\lambda_n(h)$ is reduced to that of $\mu_n(h)$.

By shrinking the spectral window, we can even get a localization with respect to  the $s$ variable, as stated in the next proposition.
\begin{proposition}\label{prop.Agmons}
	Let $M>0$ and $\eta\in (0,\frac13)$. There exist $\varepsilon, C,h_0>0$ such that, for all $h\in(0,h_0)$, and  for all eigenfunctions $\psi$ of $\widetilde{\mathscr{M}}_h$ corresponding to eigenvalues $\lambda\leq x_{\min}+h^{\frac 23}z_{1}+Mh$ we have:
\[\int_{\mathcal{W}_{\delta,h}}e^{\varepsilon s^{2}/h}|\psi|^2\dd s\dd t\leq C\|\psi\|^2\,.\]
\end{proposition}
\begin{proof}
The proof follows the same lines as that of Proposition \ref{thm.Agmon}. We let $\Phi(s)=\varepsilon s^2/2$ and write the Agmon formula:
\begin{multline*}
\int_{\mathcal{W}_{\delta,h}}\Big(m^{-2}|h\partial_se^{\Phi/h}\psi|^2+
|h\partial_te^{\Phi/ {h}}\psi|^2+{\tau_1}(s,t)|e^{\Phi/h}\psi|^2\\
-(\lambda+|\nabla\Phi|^2)|e^{\Phi/h}\psi|^2 \Big)m\dd s\dd t=0\,.
\end{multline*}	
First, we drop the tangential derivative:
\[\int_{\mathcal{W}_{\delta,h}}\left(|h\partial_te^{\Phi/{h}}\psi|^2+{\tau_1}(s,t)|e^{\Phi/h}\psi|^2-(\lambda+|\nabla\Phi|^2)|e^{\Phi/h}\psi|^2 \right)m\dd s\dd t\leq0\,.\]
We observe from \eqref{eq.c1} that there exist $\tilde{k}>0$ such that $\tau_1(s,t)\geq x_{\rm min}+t+ \tilde{k} s^2$.  Introducing the last  inequality in the above integral we get:
\begin{multline*}
\int_{\mathcal{W}_{\delta,h}}\Big(|h\partial_te^{\Phi/h}\psi|^2+t|e^{\Phi/h}\psi|^2\Big)m\dd s\dd t\\
+\int_{\mathcal{W}_{\delta,h}}(-\lambda+x_{\min}+\tilde k s^2-|\nabla\Phi|^2)|e^{\Phi/h}\psi|^2 \; m\dd s\dd t\leq 0\,.
\end{multline*}
On $\mathcal{W}_{\delta,h}$, for sufficiently small $h$ , there exists $C>0$ such that  $ m(s,t)\geq 1 -Ch^{2/3-\eta}$.
By using this in the above  integrals we have that
\begin{multline*}
(1-Ch^{2/3-\eta})\int_{\mathcal{W}_{\delta,h}}\Big(|h\partial_te^{\Phi}\psi|^2+t|e^{\Phi/h}\psi|^2\Big)\dd s\dd t\\
+\int_{\mathcal{W}_{\delta,h}}(-\lambda+x_{\min}+ \tilde k s^2-|\nabla\Phi|^2)|e^{\Phi/h}\psi|^2 \;m\dd s\dd t\leq 0\,.
\end{multline*}
Then,  with a possibly larger constant $C$ we have that
\begin{multline*}
(1-Ch^{2/3-\eta})\int_{\mathcal{W}_{\delta,h}}\Big(|h\partial_te^{\Phi}\psi|^2+t|e^{\Phi/h}\psi|^2-h^{\frac 23}z_1|e^{\Phi/h}\psi|^2\Big)\dd s\dd t\\
+\int_{\mathcal{W}_{\delta,h}}(-\lambda+x_{\min}+h^{\frac 23}z_1+\tilde k s^2-|\nabla\Phi|^2-Ch^{\frac 43- \eta})|e^{\Phi/h}\psi|^2 \;m\dd s\dd t\leq 0\,.
\end{multline*}
By using the min-max principle and the Dirichlet bracketing (only with respect to $t$, $s$ being fixed), we have
\[\int_{\mathcal{W}_{\delta,h}}\Big(|h\partial_te^{\Phi}\psi|^2+t|e^{\Phi/h}\psi|^2-h^{\frac 23}z_1|e^{\Phi/h}\psi|^2\Big)\dd s\dd t\geq 0\,.\]
Therefore,
\[\int_{\mathcal{W}_{\delta,h}}\Big((-\lambda+x_{\min}+h^{\frac 23}z_1+\tilde k s^2-|\nabla\Phi|^2-Ch^{\frac 43-{\eta}})|e^{\Phi/h}\psi|^2 \Big)m\dd s\dd t\leq 0\,,\]
so that, using the assumption on the location of $\lambda$, we obtain
\[\int_{\mathcal{W}_{\delta,h}}(-{M}h+{\tilde k}s^2-|\nabla\Phi|^2-Ch^{\frac 43-{\eta}})|e^{\Phi/h}\psi|^2 \Big)m\dd s\dd t\leq 0\,.\]
 Now if $\eta\in\left(0,\frac 13\right)$ is kept fixed and $h$ is small enough, then:
\[\int_{\mathcal{W}_{\delta,h}}(\tilde k s^2-|\nabla\Phi|^2-{ 2Mh})|e^{\Phi/h}\psi|^2 m\dd s\dd t\leq 0\,,\]
and 
\[\int_{\mathcal{W}_{\delta,h}}\Big ((\tilde k -\varepsilon^2)s^2- 2Mh\Big )|e^{\Phi/h}\psi|^2 \dd s\dd t\leq 0\,.\]
For $\varepsilon$ small enough, the conclusion follows as in the proof of Proposition \ref{thm.Agmon}.
\end{proof}

This shows that the eigenfunctions of \enquote{low energy} are localized near $A_0$ at a scale $h^{1/2}$ in the $s$ direction, and at a scale $h^{2/3}$ in the $t$ direction.

\section{Proof of the main theorem}
In view of Equation \eqref{eeq1} and Proposition \ref{prop.exp-approx} our main result  Theorem~\ref{thm.main} is a direct consequence of the next proposition that provides the asymptotic behavior of the low-lying eigenvalues, $\mu_n(h)$, of the operator   $\widetilde{\mathscr{M}}_h$ defined in the previous section. 
\begin{proposition}\label{prop-mu}
	Let $n\in \{1,2,\dots\}$. Then, as $h\to 0$ we have
	\[\mu_n(h)=x_{\min}+z_{1}h^{\frac 23}+h(2n-1)\sqrt{\frac{\kappa_0}{2}}+\mathscr{O}(h^{ \frac 43})\,.\]
	\end{proposition}
In the remainder of this section  we show the above proposition by obtaining suitable upper and lower bounds. 
\subsection{Upper bound}
To get the upper bound in Theorem \ref{thm.main}, it is sufficient to use convenient test functions in the domain of $\widetilde{\mathscr{M}}_h$ and apply the min-max principle.

Consider a smooth function $\chi$ with compact support equal to $1$ near $(0,0)$ on a scale of order $h^{1/2-\theta}$ in the $s$ direction, and on a scale of order $h^{2/3-\theta}$ in the $t$ direction, for some $\theta\in (0,\eta)$. Let us introduce the following family of test functions
\[\varphi_{n,h}(s,t)=\chi(s,t)f_{n,h}(s) a_{h}(t)\,,\]
where 
\begin{enumerate}[\rm i.]
\item $f_{n,h}(s)=h^{-\frac{1}{4} }f_n(h^{-\frac{1}{2}}s)$ and $(f_n)_{n\in\mathbb{N}\setminus\{0\}}$ is an $L^2(\R)$-normalized family of eigenfunctions of the harmonic oscillator $-{h^2}\partial^2_s+\frac{\kappa_0}{2}s^2$ (rescaled Hermite functions), 
\item $a_{h}(t)=  h^{-\frac{1}{3}} \mathrm{Ai}(h^{-\frac{2}{3}}t-z_1)$ where $\mathrm{Ai}$ is the $L^2(\R^+)$ normalized Airy function and $-z_1<0$ is its first zero. In particular, we have that $a_{h}(t)$ satisfies 
$$(-h^2\partial_t^2+t)a_h(t)= z_1  h^{2/3}  a_h(t) $$ for $t> 0$ with Dirichlet boundary conditions.
\end{enumerate}
An explicit  computation shows that  for any $n\in \mathbb{N}$ there exist constants $\varepsilon, c>0$ such that
	\begin{equation}\label{eqn1}
	\int_{\R}e^{\varepsilon s^{2}/h}|f_{n,h}(s)|^2\dd s\leq c\,,\quad 
\int_{\R^+} e^{\varepsilon t^{\frac 32}/h}|a_{h}(t)|^2\dd t\leq c\,.
\end{equation}
Using this we immediately see that we find $c_1,c_2>0$ such that
\begin{equation}\label{es-1}
1\ge  \|\varphi_{n,h}\|^2\ge 1-c_1\int t \ab{\varphi_{n,h}(s,t)}^2 \dd s\dd t
\ge 1-c_2 h^{2/3}.
\end{equation} 

We are interested in estimating the matrix elements 
$\langle\varphi_{n',h}, (\widetilde{\mathscr{M}}_h-x_{\rm min})\varphi_{n,h}\rangle.$  First, we notice that the support of $\chi-1$ and the supports of the derivatives of $\chi$ are located in a region where either $h^{-1/2}|s|\geq C h^{-\theta}$ or $h^{-2/3} t\geq C h^{-\theta} $, thus all the integrals entering the matrix element which contain such derivatives will be of order $h^{-\infty}$ due to the exponential localization of the $f_n$'s and of ${\rm Ai}$. Second, the operators $h\partial_s$ and $h\partial_t$ acting on $f_{n,h}$'s and $a_h$ respectively will generate a factor of at most $h^{1/3}$, and each integral contains two such factors; thus each term in the scalar product has an order of magnitude of at most $h^{2/3}$. Third, if we replace the function $m$ in \eqref{es-2} by $1$, the error contains an extra  factor $t$ which due to the decay of $a_h$ may be replaced by $h^{2/3}$. Together with the a-priori decay of $h^{2/3}$ coming from the derivatives, this error term will grow at most  like $h^{4/3}$.  

Moreover, replacing $\tau_1(s,t)-x_{\rm min}$ with $\frac{\kappa_0}{2}s^2+t$ will produce an error like $h^{3/2}$. Thus we may write:
\begin{align}\label{july1}
\langle\varphi_{n',h}, (\widetilde{\mathscr{M}}_h-x_{\rm min})\varphi_{n,h}\rangle=&h^2\int_\R\left ( \partial_s  f_{n,h} \overline{\partial_s  f_{n',h}}+\frac{\kappa_0}{2}s^2f_{n,h} \overline{  f_{n',h}}\right)\dd s\nonumber \\
&+ \delta_{n,n'}\int_{\R_+}\left ( \vert \partial_t  a_{h}\vert ^2 +t\vert a_h\vert^2\right)\dd t +\mathscr{O}(h^{4/3})\nonumber \\
=&\delta_{n,n'}\left  ( (2n-1)h\sqrt{\frac{\kappa_0}{2}} +z_1 h^{2/3}\right )+\mathscr{O}(h^{4/3}).
\end{align}
Using similar arguments, the Gram-Schmidt matrix elements $\langle \varphi_{n',h},\varphi_{n,h}\rangle $ will equal $\delta_{n,n'}+\mathscr{O}(h^{2/3})$. Thus if $N$ is fixed and $h$ is small enough, the subspace \[\underset{1\leq j\leq N}{\text{span}}\varphi_{j,h}\] 
will have dimension $N$ and we may find an orthonormal basis $\{\psi_{n,h}\}_{n=1}^N$ such that 
$$\psi_{n,h}=\sum_{j=1}^N c_j\varphi_{j,h},\quad c_j=\delta_{n,j}+\mathscr{O}(h^{2/3})\,.$$

Thus the matrix elements $\langle\psi_{n',h}, (\widetilde{\mathscr{M}}_h-x_{\rm min})\psi_{n,h}\rangle$
will obey the same estimate as in \eqref{july1}, which via the min-max principle imply 
\begin{align*}
 \mu_n(h)\le x_{\min}+z_{1}h^{\frac 23}+h(2n-1)\sqrt{\frac{\kappa_0}{2}}+\mathscr{O}(h^{ \frac 43}),\quad 1\leq n\leq N.
\end{align*}
\subsection{Lower bound}
Let $N\geq 1$ and consider a family of eigenfunctions $(\psi_{j,h})_{1\leq j\leq N}$ associated with the eigenvalues $(\mu_j(h))_{1\leq j\leq N}$. 
We let
\[\mathscr{E}_N(h)=\underset{1\leq j\leq N}{\mathrm{span}}\, \psi_{j,h}\,\subset L^2(\mathcal{W}_{\delta,h}, m\dd s\dd t)\,.\]
Note that the decay estimates of Propositions \ref{prop.Agmont} and \ref{prop.Agmons} can be extended to $\psi\in \mathscr{E}_N(h)$.

 Let us choose any $\psi\in \mathscr{E}_N(h)$ with norm one. Because $m(s,t)\leq 1$, we have the important inequality
\begin{equation}\label{hc-8}
   1= \int_{\mathcal{W}_{\delta,h}}|\psi|^2 m\dd s\dd t\leq \Vert \psi\Vert^2_{L^2(\mathcal{W}_{\delta,h}; \dd s\dd t)}\,.
\end{equation}
We also have
\begin{equation}\label{hc-9}
\mu_N(h)\geq \langle  \widetilde{\mathscr{M}}_h\psi,\psi\rangle=\int_{\mathcal{W}_{\delta,h}}\left( m|h\partial_t\psi|^2+m^{-1}|h\partial_s\psi|^2+m {\tau}_1(s,t)|\psi|^2\right)\dd s\dd t\,.
\end{equation}
We recall that $m(s,t)=1-\kappa(s)t$ so that,  using \eqref{eq.Agmont2} in the second inequality below,
\[\begin{split}
\langle { \widetilde{\mathscr{M}}_h}\psi,\psi\rangle&\geq\int_{\mathcal{W}_{\delta,h}}\left( |h\partial_t\psi|^2+|h\partial_s\psi|^2+m {\tau}_1(s,t)|\psi|^2\right)\dd s\dd t-C\int_{\mathcal{W}_{\delta,h}}t|h\nabla_{s,t}\psi|^2\dd s\dd t\\
&\geq\int_{\mathcal{W}_{\delta,h}}\left( |h\partial_t\psi|^2+|h\partial_s\psi|^2+m {\tau}_1(s,t)|\psi|^2\right)\dd s\dd t-Ch^{4/3}\,,\\
&=x_{\rm min}+\int_{\mathcal{W}_{\delta,h}}\left( |h\partial_t\psi|^2+|h\partial_s\psi|^2+m {(\tau_1(s,t)-x_{\rm min})}|\psi|^2\right)\dd s\dd t-Ch^{4/3}\\
&{\ge x_{\min}+\int_{\mathcal{W}_{\delta,h}}\left( |h\partial_t\psi|^2+|h\partial_s\psi|^2+\left(t+\kappa_0\frac{s^2}{2}\right)|\psi|^2\right)\dd s\dd t-Ch^{4/3}}\\
&\qquad
{-\tilde{C} \int_{\mathcal{W}_{\delta,h}}(|s^2t|+|t^2|)|\psi|^2\dd s\dd t },
\end{split}\]
the last integral can be estimated to be of  order $h^{4/3}$ as well using the exponential decay in the $s$ and $t$ variables (Propositions \ref{prop.Agmont} and \ref{prop.Agmons}). Then, there is a $C>0$ such that
\begin{multline*}
\langle  \widetilde{\mathscr{M}}_h\psi,\psi\rangle\geq x_{\min}+\int_{\mathcal{W}_{\delta,h}}\left( |h\partial_t\psi|^2+|h\partial_s\psi|^2+\left(t+\kappa_0\frac{s^2}{2}\right)|\psi|^2\right)\dd s\dd t
 -C{h^{4/3}}\,.
\end{multline*}
 Now using the inequalities in \eqref{hc-9} and \eqref{hc-8} we have:
\begin{align}\label{hc-100}
\int_{\mathcal{W}_{\delta,h}}\left( |h\partial_t\psi|^2+|h\partial_s\psi|^2+\left(t+\kappa_0\frac{s^2}{2}\right)|\psi|^2\right)&\dd s\dd t \leq \mu_N(h)-x_{\min}+{Ch^{ 4/3}}
\\
&\leq (\mu_N(h)-x_{\min}+{Ch^{4/3}})\|\psi\|^2_{L^2(\dd s \dd t)}\,.\nonumber 
\end{align}
On the left hand side of the above identity we recognize the quadratic form associated to the operator
 $$\mathscr{A}_h\otimes \mathbf{1}+ \mathbf{1}\otimes \mathscr{H}_h\,,$$ 
  where  $\mathscr{A}_h=-h^2\partial^2_t+t$ on $L^2((0,+\infty), \dd t)$ with Dirichlet boundary conditions and  $\mathscr{H}_h=-h^2\partial^2_s+\kappa_0\frac{s^2}{2}$ on 
 $L^2((-\infty,\infty), \dd s)$.  The spectrum of $\mathscr{A}_h$ and $ \mathscr{H}_h$
 is given by 
  \begin{equation*}
 \mathrm{sp}(\mathscr{A}_h)=\Big\{z_{n}h^{2/3}\,,\, n\ge 1\Big\}\,,\qquad
  \mathrm{sp}(\mathscr{H}_h)=\Big\{(2n-1)h\sqrt{\frac{\kappa_0}{2}}\,,\,
   n\ge 1\Big\}\,.
 \end{equation*}
 Thus, for $h$ small enough, the $N$-th eigenvalue of $\mathscr{A}_h\otimes \mathbf{1}+ \mathbf{1}\otimes \mathscr{H}_h$, denoted by $\nu_N(h)$, is given by 
 \begin{equation}\label{eq.lnAH}
 \nu_N(h)=z_{1}h^{2/3}+(2N-1)h\sqrt{\frac{\kappa_0}{2}}\,.
 \end{equation}
 Notice that the  set $\mathscr{E}_N(h)$ is contained in the form domain of $\mathscr{A}_h\otimes \mathbf{1}+ \mathbf{1}\otimes \mathscr{H}_h$ and, seen as a subset of $L^2(\mathcal{W}_{\delta,h};\dd s\dd t)$, still has the dimension $N$ for  $h$ small enough. This is because the $\psi_{j,h}$'s are almost orthogonal in the \enquote{flat} space, up to an error of order $h^{2/3}$. Thus, by the min-max 
 principle, we have 
 \begin{align}\label{es1}
  \nu_N(h)\le \sup_{\psi\in \mathscr{E}_N(h)} \frac{\left \langle \Big (\mathscr{A}_h\otimes \mathbf{1}+ \mathbf{1}\otimes \mathscr{H}_h \Big )\psi, \psi\right \rangle_{L^2(dsdt)}}{\|\psi\|^2_{L^2(\dd s \dd t)}}
  \le  (\mu_N(h)-x_{\min}+{Ch^{4/3}})\,,
 \end{align} 
 where in the last inequality we used  \eqref{hc-100}.  This implies the desired lower bound and concludes the proof of Proposition~\ref{prop-mu}.
 
\subsection*{Acknowledgment}
D.K. has been supported by the EXPRO grant No. 20-17749X
of the Czech Science Foundation (GACR). D.K. is also grateful to the Aalborg University for supporting his stay in 2019. E.S. has been partially funded by Fondecyt (Chile) project  \# 118--0355 and thanks Esteban Ramos Moore for stimulating discussions. H.C., E.S, and N.R. are deeply grateful to the Mittag-Leffler Institute where this collaboration was stimulated during the thematic semester \enquote{Spectral Methods in Mathematical Physics} in 2019. E.S. and N.R. are also grateful to the CIRM where the ideas of this paper were discussed (\enquote{Research in Pairs} session in 2019).

\bibliographystyle{abbrv}
\bibliography{stark}

\end{document}